\documentclass{article}

\usepackage{arxiv}

\usepackage{amsmath,amsfonts,amssymb, amsthm}

\usepackage[usenames,dvipsnames]{xcolor}
\definecolor{niceRed}{RGB}{190,38,38}
\definecolor{niceBlue}{HTML}{0466a7}

\usepackage[utf8]{inputenc} 
\usepackage[T1]{fontenc}    
\usepackage{hyperref}       
\hypersetup{
	colorlinks = true,
	urlcolor = {black},
	linkcolor = {niceBlue},
	citecolor = {niceRed} 
}

\usepackage{natbib}
\bibpunct[, ]{(}{)}{,}{a}{}{,}%
\def\BIBand{and}%

\usepackage{url}            
\usepackage{booktabs}       
\usepackage{amsfonts}       
\usepackage{nicefrac}       
\usepackage{microtype}      
\usepackage{cleveref}       
\usepackage{lipsum}         
\usepackage{graphicx}
\usepackage{natbib}
\usepackage{doi}

\usepackage{mathtools}
\renewcommand*\vec[1]{\boldsymbol{#1}}
\DeclareMathOperator{\EXP}{\mathbb{E}}

\newcommand*{\diff}{\textup{d}}
\renewcommand*\vec[1]{\boldsymbol{#1}}

\DeclareMathOperator*{\argmin}{arg\,min}

\usepackage{algorithm,algpseudocode}

\newtheorem{lemma}{Lemma}[section]
\newtheorem{problem}{Problem}

\newtheorem{proposition}{Proposition}[section]

\newtheorem{remark}{Remark}[section]

\title{On the Convergence of Learning Algorithms in Bayesian Auction Games}

\newif\ifuniqueAffiliation
\uniqueAffiliationtrue

\ifuniqueAffiliation 
\author{
	Martin Bichler \\
	\small Department of Computer Science\\
	\small Technical University of Munich\\
	\small \texttt{bichler@cit.tum.de} \\
	\And
	Stephan B. Lunowa\\
	\small Department of Mathematics\\
	\small Technical University of Munich \\
	\small  \texttt{stephan.lunowa@tum.de} \\
	\AND
	Matthias Oberlechner\\
	\small Department of Computer Science\\
	\small Technical University of Munich \\
	\small \texttt{matthias.oberlechner@tum.de} \\
	\And
	Fabian R. Pieroth\\
	\small Department of Computer Science\\
	\small Technical University of Munich \\
	\small \texttt{fabian.pieroth@tum.de} \\
	\And 
	Barbara Wohlmuth \\
	\small Department of Mathematics\\
	\small Technical University of Munich \\
	\small \texttt{wohlmuth@tum.de} \\
}
\else
\usepackage{authblk}

\newbox{\orcid}\sbox{\orcid}{\includegraphics[scale=0.06]{orcid.pdf}} 
\author[1]{%
	\href{https://orcid.org/0000-0000-0000-0000}{\usebox{\orcid}\hspace{1mm}David S.~Hippocampus\thanks{\texttt{hippo@cs.cranberry-lemon.edu}}}%
}
\author[1,2]{%
	\href{https://orcid.org/0000-0000-0000-0000}{\usebox{\orcid}\hspace{1mm}Elias D.~Striatum\thanks{\texttt{stariate@ee.mount-sheikh.edu}}}%
}
\affil[1]{Department of Computer Science, Cranberry-Lemon University, Pittsburgh, PA 15213}
\affil[2]{Department of Electrical Engineering, Mount-Sheikh University, Santa Narimana, Levand}
\fi

\renewcommand{\headeright}{}
\renewcommand{\undertitle}{}
\renewcommand{\shorttitle}{}

\hypersetup{
pdftitle={On the Convergence of Learning Algorithms in Bayesian Auction Games},
pdfsubject={cs.GT},
pdfauthor={Martin~Bichler, 	Stephan B.~Lunowa, Matthias~Oberlechner, Fabian R.~Pieroth, Barbara~Wohlmuth},
pdfkeywords={{Bayes--Nash} equilibrium, variational inequality, gradient-flow algorithms} ,
}

\begin{document}
\maketitle

\begin{abstract}
	Equilibrium problems in Bayesian auction games can be described as systems of differential equations. Depending on the model assumptions, these equations might be such that we do not have a rigorous mathematical solution theory. The lack of analytical or numerical techniques with guaranteed convergence for the equilibrium problem has plagued the field and limited equilibrium analysis to rather simple auction models such as single-object auctions.  
	Recent advances in equilibrium learning led to algorithms that find equilibrium under a wide variety of model assumptions. We analyze first- and second-price auctions where simple learning algorithms converge to an equilibrium. The equilibrium problem in auctions is equivalent to solving an infinite-dimensional variational inequality (VI). 
	Monotonicity and the Minty condition are the central sufficient conditions for learning algorithms to converge to an equilibrium in such VIs. We show that neither monotonicity nor pseudo- or quasi-monotonicity holds for the respective VIs. The second-price auction's equilibrium is a Minty-type solution, but the first-price auction is not.
	However, the Bayes--Nash equilibrium is the unique solution to the VI within the class of uniformly increasing bid functions, which ensures that gradient-based algorithms attain the equilibrium in case of convergence, as also observed in numerical experiments.
\end{abstract}


\section{Introduction}

The Nash equilibrium is the central solution concept in non-cooperative game theory \citep{nash1950equilibrium}.  
Nash's original idea was extended to games with incomplete information \citep{harsanyi1967games}. Today, the analysis of Bayes--Nash equilibrium (BNE) is key to the study of auctions, contests, and many other economic models, where the strategic complexity for participants arises from their uncertainty about the preferences of competitors \citep{vickrey1961counterspeculation}. 

Auction theory studies allocation and prices on markets with self-interested participants in equilibrium. Such predictions are important because most market institutions in this field do not have simple dominant strategies for bidding truthfully. The celebrated Vickrey-Clarke-Groves (VCG) mechanism was shown to be unique in this respect in the widely used standard model with independent private values and quasi-linear (i.e., payoff maximizing) bidders \citep{green1977CharacterizationSatisfactoryMechanisms}.
However, for various reasons, the VCG mechanism is rarely used in practice \citep{ausubel2006LovelyLonelyVickrey}. Therefore, it is essential to understand non-truthful equilibrium strategies in markets more generally. 

Participants in auction games do not have complete but only distributional information about the valuations of competing market participants. The value of competing bidders is modeled as a draw from an atomless distribution function. An equilibrium bid function then determines how much they bid based on their value draw and their knowledge of the prior distribution. The first-price sealed-bid auction in the symmetric and independent prior values model provides a textbook example, where the first-order conditions lead to an ordinary differential equation that admits a closed-form solution for the BNE bid function \citep{krishna2009AuctionTheory}.

Unfortunately, despite the enormous academic attention auction theory has received, BNE strategies are only known for very restricted, simple market models. For bidders with non-uniform or interdependent valuation distributions, multiple objects, or non-quasilinear (i.e., non-payoff-maximizing) utility functions, we typically do not know an explicit equilibrium bid function. 
The equilibrium problem often leads to a system of non-linear differential equations for which we have no exact mathematical solution theory. 
Beyond analytical solutions, numerical techniques for solving differential equations have turned out to be challenging and have received only limited attention. \citet{fibich2011numerical} criticize the inherent numerical instability of standard techniques used in this field. 

\subsection{Equilibrium learning}

Equilibrium learning is an alternative numerical approach to finding equilibrium, as compared to standard techniques for solving differential equations. This literature examines what kind of equilibrium arises as a consequence of {a relatively simple process of learning and adaptation}, in which agents are trying to maximize their payoff while learning about the actions of other agents \citep{fudenberg1998theory, hart2003uncoupled, young2004strategic}. 
Learning provides an intuitive, tractable model of how equilibrium emerges but it is well established that independent learning dynamics do not generally obtain a Nash equilibrium \citep{benaimMixedEquilibriaDynamical1999}.
Numerical analyses of matrix games show that gradient-based algorithms can circle, diverge, or they are even chaotic \citep{sanders2018prevalence, bielawski2021follow, chotibut2020route, palaiopanos2017multiplicative, vlatakisgkaragkounis2023chaos}. 
Recently, \citet{milionis2022nash} proved that there exist games, for which all game dynamics fail to converge to a Nash equilibrium from all starting points. On the other hand, there are also game classes such as potential games or games with strictly dominated strategies, where learning algorithms do converge. 

\subsection{Learning in auctions}

In recent years, a number of learning algorithms were introduced for Bayesian auction games with a continuous type and action space, and they showed convergence on a wide variety of auction models ranging from simple single-object auctions in the independent private values model to interdependent valuations and models with multiple objects \citep{bichler2021npga, bichler2023learning}.  
Equilibrium can be verified ex-post, but the reasons for the convergence of such learning algorithms in auction games have been a mystery so far. If auctions are indeed learnable, this has significant consequences for theory and applications alike. First, numerical solvers could be provided for models that could not be solved analytically so far. Second, if already stylized repeated auctions were not learnable, what would this mean for the growing economy of autobidding agents in display ad auctions and related applications? 

What is so different in auction games compared to general-sum games, such that we see convergence in such a wide variety of models?
Recently, it has been shown that a large class of learning algorithms converges in auction games with complete information, where the valuations of bidders are known a priori \citep{kolumbus2022auctions}. However, the challenge in auctions is that we have incomplete information and Bayesian models are quite different. We focus on the reasons for the convergence of learning algorithms in incomplete-information auction games.

For this, we draw on the field of operator theory and variational inequalities. 
Every Nash Equilibrium (NE) can be seen as a solution to a Stampacchia-type variational inequality (VI), and in some cases, the reverse is also true, for example with quasi-concave utility functions. This connection also holds for auction games and infinite-dimensional VIs  \citep{cavazzutiNashEquilibriaVariational2002}. Interestingly, the link between Nash equilibria and VIs has been explored for traffic games \citep{patriksson2003sensitivity} or Walrasian equilibrium \citep{jofre2007variational}, but not for Bayesian games with continuous type and action space as is the case in auction theory. 
In contrast to finite games, auctions need to be modeled as infinite-dimensional variational inequalities.

In the literature on variational inequalities, two sufficient conditions are known, for which some types of algorithms always converge to an equilibrium. They can be seen as a generalization of convexity in optimization. 
The \textit{monotonicity} condition is the most well-known condition to guarantee convergence for VIs \citep{bauschkeConvexAnalysisMonotone2017}.
Various first-order projection methods, as discussed by \citet{tseng1995linear}, converge to a unique solution of a monotone VI,
and higher-order methods have also been developed \citep{adil2022optimal, lin2022perseus}.
Monotonicity is also central for convergence guarantees in the recent literature on learning in games \citep{ratliff2013characterization, chasnov2019convergence}. Apart from this, the Minty condition has drawn some attention, as extragradient algorithms are known to converge to equilibrium if this condition holds globally \citep{strodiot2016class, Song2020OptDualExtrapolation}. This condition is also referred to as the {Minty VI or dual VI} of the Stampaccia-type VI \citep{ye2022infeasible}. In constant-sum games it is related to the celebrated smoothness condition of games \citep{anagnostides2023interplay} introduced by \citep{roughgarden2015intrinsic}.
It is natural to ask if the monotonicity and the Minty condition, two sufficient conditions, hold in auction games.

\subsection{Contributions}

Demonstrating the monotonicity of auction games is challenging. Whereas practical equilibrium learning algorithms employ some form of discretization, previous research has shown that in discretized versions of auction games, the monotonicity condition is sometimes violated \citep{soda2022}. However, such non-monotonicities could arise due to the game's discretization. Violations in a discretized game might vanish in games with continuous types and actions. 
That is why we study whether monotonicity or the Minty condition is satisfied in infinite dimensions in a function space. If any of the two conditions were satisfied in function space, this would explain the convergence of algorithms also in discretized versions of the game, where the condition is violated \citep{Glowinski1981NumericalAnalysisVI}.

First, we recover the well-known symmetric equilibrium strategies for the first-price and the second-price sealed-bid auction in the symmetric independent private-values model \citep{krishna2009AuctionTheory}, but with a new proof technique based on the Gateaux derivative of the ex-ante utility function, which is novel and useful in its own right. 
Second, this proof technique for equilibrium problems in auctions and the resulting operator for the Gateaux derivative allows us to analyze the monotonicity and the Minty conditions. 

Our findings reveal that the first- and the second-price auctions are neither monotone nor pseudo- or quasi-monotone.
Thus, we look at the Minty condition for variational inequalities.   
While the dominant-strategy incentive-compatible second-price auction satisfies this condition, this is not the case for the first-price auction. 
However, an ex-post guarantee is provided by showing the uniqueness of the VI solution (the BNE) within the class of uniformly increasing bid functions, which ensures that every gradient-based learning algorithm must converge to this solution if it does converge (within this function class).
Our numerical analyses in the space of piecewise-linear functions illustrate the reasons for the Minty violations and show that they are without loss and that the gradient flows lead to the BNE in the first-price auction independent of the starting point and independent of Minty violations. 

\section{Related Literature}

More than 60 years ago, \citet{vickrey1961counterspeculation} derived the Bayes--Nash equilibrium (BNE) strategy in a single-object first-price auction in the independent-private values model with symmetric bidders, uniform prior distributions, and quasilinear utility functions. The first-order conditions, together with the assumption of symmetric bidding behavior, lead to a linear ordinary differential equation, which has a closed-form solution for the BNE bidding strategy. Unfortunately, this is only the case for the symmetric and independent private values model. 
It turns out that deviations from this benchmark model lead to challenges in the equilibrium analysis. For example, even in the asymmetric independent private values model, where different bidders' values are drawn from different distributions, we get a system of non-linear differential equations; and no closed-form expression for the bidding strategies exist for general distributions \citep{hubbard2014ChapterNumericalSolution}. 

The first question is whether BNE always \textit{exists} in such auction games. For finite, complete-information games, we know of the existence of a mixed Nash equilibrium \citep{nash1950equilibrium}. \citet{glicksberg1952further} extended the existence result to games with continuous and compact action sets. For Bayesian games with continuous action space, \citet{jackson2005existence} provide assumptions for the existence of equilibrium in distributional strategies. For example, first-price and second-price single-unit auctions, all-pay auctions, double auctions, and multi-unit discriminatory or uniform price auctions were shown to have an equilibrium in distributional strategies, not necessarily in pure strategies. 
It is interesting to note that there are auction models where there is no BNE of the continuous game, but there are equilibria in the discretized game \citep{jackson2002communication}. 
While there was significant progress, we do not have a complete theory when a BNE exists in general continuous-type and -action auction games \citep{carbonell-nicolau2018ExistenceNashEquilibrium}. 
Note that if {one could} prove the strict monotonicity of the game, this also proves the uniqueness of the equilibrium in the respective auction. Therefore, we can already deduce that strict monotonicity cannot hold for the second-price auction considering asymmetric strategies, where several equilibria are known to exist \citep[p. 118]{krishna2009AuctionTheory}. However, we know that it is unique in the symmetric case.

Another question concerns the \textit{computational complexity} in those cases where Nash equilibria do exist. 
Computing a Nash equilibrium in a finite, complete-information game is generally PPAD-hard \citep{daskalakisComplexityComputingNash2009}. 
\citet{cai2014simultaneous} showed that finding an exact BNE in specific simultaneous auctions for multiple items is at least hard for PP, with PP being a complexity class higher than the polynomial hierarchy and close to PSPACE. Such problems can be considered intractable even for very small problem sizes. We know little about the complexity of finding BNE in other multi-item auctions. 
Recently, \citep{chen2023complexity} proved that there is a PTAS for approximating first-price auctions with a uniform tie-breaking rule, but that the approximation problem is PPAD-complete for other tie-breaking rules. 

A related stream in the literature focuses on \textit{learning in games} \citep{fudenberg1998theory}. For example, fictitious play is a natural method by which agents iteratively search for a pure Nash equilibrium and play the best response to the empirical frequency of play of other players \citep{brownIterativeSolutionGames1951}. Several algorithms have been proposed based on best or better response dynamics for finite and simultaneous-move games, and the literature is large \citep{abreu1988structure, hart2000simple, fudenberg1998theory, hart2003uncoupled, young2004strategic}. 
Learning dynamics do not always converge to equilibrium \citep{daskalakis2010learning, vlatakis2020no}. Learning algorithms can cycle, diverge, or be chaotic; even in zero-sum games, where the Nash equilibrium is tractable \citep{mertikopoulos2018cycles, baileyMultiplicativeWeightsUpdate2018, cheung2020chaos}. \citet{sanders2018prevalence} argues that chaos is typical behavior for more general large matrix games. 
Recent results have shown that learning dynamics do not converge in games with mixed Nash equilibria \citep{baileyMultiplicativeWeightsUpdate2018, letcher2019differentiable}, and that there are finite games for which any dynamics is bound to have starting points that do not end up at a
Nash equilibrium \citep{milionis2022nash}. 
Overall, the dynamics of matrix games can be arbitrarily complex and hard to characterize \citep{andrade2021learning}. 

It is well-known that no-external-regret learners converge to a (Bayesian) coarse correlated equilibrium, but the convergence of learning algorithms to a Nash equilibrium is only known for specific types of games such as finite potential games \citep{foster1997calibrated, hart2000simple,jafari2001no,stoltz2007learning,hartlineNoRegretLearningBayesian2015,foster2016learning}.
In general, the convergence of no-regret learning algorithms to a Nash equilibrium depends on the specific structure of the game and the algorithm used. The convergence of learning algorithms to a Nash equilibrium in Bayesian auction games with continuous type and action space is largely unexplored.


We draw on the literature on \textit{variational inequalities} (VI), in particular infinite-dimensional VIs.  
Variational inequalities were initially introduced to study equilibrium problems in physical systems \citep{zbMATH03235902,Kikuchi1988}. Since then, the theory's applicability has been expanded to include various problems \citep{lionsVariationalInequalities1967, Glowinski1981NumericalAnalysisVI, Kinderlehrer2000VI}.
Notably, the search for equilibrium in a game can be reformulated as solving a related VI. 
Auction models with continuous type and action space can be described as infinite-dimensional variational inequalities \citep{cavazzutiNashEquilibriaVariational2002}. 

Monotonicity can be seen as a generalization of convexity in optimization, and it is the primary property to show the convergence of algorithms in VIs. For first-order methods \citet{nemirovski2004prox} proved that the extragradient method converges to a weak solution with a global rate of $O(1/t)$ if the operator $F$ is monotone and Lipschitz-continuous. Other methods achieve the same goal \citep{nesterov2007dual, kotsalis2022simple} and achieve the lower bound of \citet{ouyang2021lower}. These are also known as ''projection-type'' methods. More recently, also higher-order methods were developed \citep{lin2022perseus}. There is also more recent work on non-monotone VIs, which usually assumes the existence of Minty-type solutions to the VI \citep{Song2020OptDualExtrapolation, huang2023beyond}. \citet{huang2023beyond} show that the existence of Minty solutions turns out to be critical in establishing convergence for the projection-type methods for non-monotone VIs. 

\section{Preliminaries} 

This section lays the foundation for studying the equilibrium problem and its associated variational inequality (VI) in a function space. To begin our analysis, it is crucial to establish a derivative in function space. This requires us to work with a set of strategies that exhibit sufficient well-behaved properties.
Furthermore, we narrow our focus to the symmetric setting with symmetric priors and strategies and the independent private values model. This choice simplifies our analysis and is sufficient to give answers to whether forms of monotonicity are the reason for the convergence of first-order methods in auction models.

\subsection{Abstract setting}

Let $n \in \mathbb{N}$ be the number of bidders. For bidder $i$, the set of possible bids is called $B_i \subset \mathbb{R}$, and the set of valuations of bidder $i$ is $\mathcal{X}_i \subset \mathbb{R}$. We define $\vec{B} := \bigtimes_{i=1}^n B_i$ and $\vec{\mathcal{X}} := \bigtimes_{i=1}^n \mathcal{X}_i$. The goal of each bidder is to maximize their payoff, i.e., they consider their utility function
\[
u_i : \vec{B} \times \vec{\mathcal{X}} \to \mathbb{R} : u_i(\vec{b}, \vec{x}) .
\]
For this, they search for a strategy $\beta_i : \vec{\mathcal{X}} \to B_i$. Let the vector space $V_i$ contain all possible strategies $\beta_i$, while the subset $\mathcal{B}_i \subset V_i$ contains all admissible strategies, and we define $\vec{V} := \bigtimes_{i=1}^n V_i$ and $\vec{\mathcal{B}} := \bigtimes_{i=1}^n \mathcal{B}_i$.
The random values $X_i$ of all bidders $i=1,\dots,n$ are distributed in $\mathcal{X}_i$ according to an atomless probability distribution $F_i$ with full support over $\mathcal{X}_i$. We denote by $U_i : \vec{V} \to \mathbb{R} : U_i(\vec{\beta}) := \EXP_{\vec{X}}[ u_i(\vec{\beta}(\vec{X}), \vec{X}) ]$ the expected utility of bidder $i$ for given strategies $\vec{\beta}$.
Note that here and in the following, we denote by $\vec{\beta}$ a vector $(\beta_1, \beta_2, \dots, \beta_n)$ and by $\beta_i, \vec{\beta}^*_{-i}$ the vector $(\beta_1^*, \dots, \beta_{i-1}^*, \beta_i, \beta_{i+1}^*, \dots, \beta_n^*)$.
The Bayesian Nash equilibrium (BNE) for this auction game is then given by:
\begin{problem}[BNE]
	Find $\vec{\beta}^* = (\beta_1^*, \dots, \beta_n^*) \in \vec{\mathcal{B}}$ such that for all $i=1,\dots,n$ there holds
	\begin{align}\label{ineq:BNE}
		U_i(\vec{\beta}^*) \geq U_i(\beta_i, \vec{\beta}^*_{-i}) \qquad \forall \beta_i \in \mathcal{B}_i .
	\end{align}
\end{problem}

Under certain conditions (to be elaborated in the following), the equilibrium condition can be reformulated as a variational inequality. For this, we need to consider the expected utility function's derivative in a function space. This demands some regularity of the underlying function space. Assume $V_i$ is a Banach space and denote by $V_i^* := \mathcal{L}(V_i, \mathbb{R})$ its dual space consisting of all continuous linear functionals on $V_i$. Note that not all linear functionals are continuous for general infinite-dimensional spaces. Furthermore, assume that $\mathcal{B}_i \subset V_i$ is convex and closed. Following the standard procedure in the literature \citep{lionsVariationalInequalities1967, Kinderlehrer2000VI}, we derive the so-called Gateaux-derivative of $U_i$, which can be understood as the generalization of the (linear) directional derivative in normed spaces.
So, let $DU_i(\vec{\beta})[d]$ denote the directional derivative of $U_i$ at $\vec{\beta} = (\beta_1, \dots, \beta_n) \in \vec{\mathcal{B}}$ with respect to $\beta_i$ along $d \in V_i$, i.e.,
\begin{align}\label{eq:Gateaux}
	DU_i(\vec{\beta})[d] := \lim_{\varepsilon\to 0} \varepsilon^{-1} \big(U_i(\beta_i + \varepsilon d, \vec{\beta}_{-i}) - U_i(\vec{\beta}) \big)  \qquad \forall d \in V_i .
\end{align}
The directional derivative is the Gateaux-derivative iff $DU_i(\vec{\beta}) \in V_i^*$, i. e., when the derivative is a continuous linear operator in the direction $d \in V_i$.
If the Gateaux-derivative exists, a \emph{necessary condition} for a BNE is the following (Stampacchia-type) VI \citep{Kinderlehrer2000VI}:
\begin{problem}[VI]
	Find $\vec{\beta}^* \in \vec{\mathcal{B}}$ such that for all $i=1,\dots,n$ there holds
	\begin{align}\label{ineq:VI}
		DU_i(\vec{\beta}^*)[\beta_i - \beta^*_i] \leq 0 \qquad \forall \beta_i \in \mathcal{B}_i .
	\end{align}
\end{problem}
A \emph{sufficient condition}, also referred to as the dual VI of the Stampacchia formulation, is given by the (Minty-type) VI:
\begin{problem}[MVI]
	Find $\vec{\beta^*} \in \vec{\mathcal{B}}$ which satisfies for all $i=1,\dots,n$
	\begin{align}\label{ineq:minty}
		D_i(\vec{\beta})[\beta_i - \beta^*_i] \leq 0 \qquad \forall \vec{\beta} \in \vec{\mathcal{B}}.
	\end{align}
\end{problem}
In general, solutions of the Minty-type VI \eqref{ineq:minty} are a subset of the BNEs given by \eqref{ineq:BNE}, which are, in turn, a subset of solutions of the Stampacchia-type VI \eqref{ineq:VI}.
Vice versa, solutions of the Stampacchia-type VI \eqref{ineq:VI} are also BNEs if $U_i$ is pseudoconvex in $\beta_i$ for all $\vec{\beta}_{-i}$, and BNEs are, in turn, solutions of the Minty-type VI \eqref{ineq:minty} if $U_i$ is quasiconvex in $\beta_i$ for all $\vec{\beta}_{-i}$ \citep{cavazzutiNashEquilibriaVariational2002}.

\subsection{Symmetric and independent private value auctions}
In the following sections, we consider second- and first-price sealed-bid auctions under the assumption of (complete) symmetry and identically independently distributed private values. (Complete) symmetry implies
\[
B_i = B,\quad \mathcal{X}_i = \mathcal{X}, \quad X_i \sim_{iid.} F_i \equiv F, \quad u_i \equiv u \qquad \forall i = 1,\dots,n.
\]
Furthermore, private values ensure $\beta_i(\vec{x}) = \beta_i(x_i)$ for $i = 1,\dots,n$, i.e., the strategy of each bidder $i$ depends only on the knowledge of their own valuation $X_i = x_i$. The ex-ante utility is denoted $U(\vec{\beta}) := \EXP_{\vec{X}}[u(\vec{\beta}(\vec{X}), \vec{X})]$, and symmetric strategies are denoted $\vec{\tilde{\beta}} := (\tilde{\beta}, \dots, \tilde{\beta}) \in \vec{\mathcal{B}}$ for $\tilde{\beta} \in \mathcal{B}$.

In the following, we assume $\mathcal{X} = [0, 1]$ (without loss of generality) and $F \in C^{0,1}([0,1])$, i.e., the cumulative probability function is Lipschitz-continuous. 
To analyze the VI we have to define an appropriate set of admissible strategies. This set should be sufficiently general to allow for strategies that may be considered as sensible for the underlying problem. Additionally, it needs to provide adequate structure, e.g., an inner product or a natural dual product. Therefore, consider the Banach space
\[
V_i = V := W^{1,1}(0,1; F) = \{ \beta \in L^1(0,1; F) \ | \ \beta' \in L^1(0,1; F) \} ,
\]
i.e., $V$ consists of $F$-integrable functions with $F$-integrable weak derivatives. Note that $V \subseteq AC([0,1])$, where the latter space denotes all absolutely continuous functions on $[0, 1]$. For small $\delta > 0$ we define
\[
\mathcal{B}_\delta := \big\{ \beta \in V\ :\ 0 \leq \beta \leq 1 \text{ $F$-a.e., } 0 < \delta \leq \beta' \text{ $F$-a.e., and } {\beta(0) = 0} \big\} .
\]
Note that the restriction $0 \leq \beta \leq 1$ is natural because only positive bids are feasible, and bidding more than the maximal valuation $1$ implies a non-positive payoff. Similarly, it is natural to assume the bids $\beta$ to be increasing in valuation. Requiring a small positive derivative ($\beta' \geq \delta > 0$) is slightly more restrictive, but together with the other assumptions, this ensures the set $\mathcal{B}_\delta$ to be convex, closed, and bounded in $V$.

In this setting, the BNE \eqref{ineq:BNE}, VI \eqref{ineq:VI} and MVI \eqref{ineq:minty} simplify to {deviations in }a single strategy:
\begin{problem}[Symmetric BNE, VI and MVI]
	A symmetric BNE $\beta^* \in \mathcal{B}_\delta$ satisfies (with $\vec{\beta}^* = (\beta^*, \dots, \beta^*)$)
	\begin{align}\label{ineq:symBNE}
		U(\vec{\beta}^*) \geq U(\beta, \vec{\beta}_{-1}^*), \qquad \forall \beta \in \mathcal{B}_\delta .
	\end{align}
	A solution $\beta^* \in \mathcal{B}_\delta$ to the symmetric VI satisfies (with $\vec{\beta}^* = (\beta^*, \dots, \beta^*)$)
	\begin{align}\label{ineq:symVI}
		DU(\vec{\beta}^*)[\beta - \beta^*] \leq 0 \qquad \forall \beta \in \mathcal{B}_\delta.
	\end{align}
	A solution $\beta^* \in \mathcal{B}_\delta$ to the symmetric MVI satisfies (with $\vec{\tilde{\beta}} = (\tilde{\beta}, \dots, \tilde{\beta})$)
	\begin{align}\label{ineq:symMVI}
		DU(\beta, \vec{\tilde{\beta}}_{-1})[\beta - \beta^*] \leq 0 \qquad \forall \beta, \tilde{\beta} \in \mathcal{B}_\delta.
	\end{align}
\end{problem}
Here $DU(\vec{\beta})[d]$ is the Gateaux-derivative of $U$ at $\vec{\beta} \in \vec{\mathcal{B}}_\delta$ with respect to $\beta_1$ along $d \in V$.

In the following two sections, we use these mathematical tools to analyze the second- and first-price sealed-bid auction in the continuous setting. Thereto, we first derive the Gateaux-derivative of the bidder's utility function, and use it to analyze the BNE, VI and MVI and their potentially different solutions in detail for these applications. Subsequently, we analyze whether the Gateaux-derivative is (quasi-)monotone, since this property would ensure convergence of (certain) gradient-based learning algorithms. In particular, we discuss counterexamples for the simplest case of two bidders with independent uniform priors.

\section{Second-price sealed-bid auction}

For symmetric second-price sealed-bid auctions with risk-neutral bidders, the utility function of a bidder is given by
\begin{align*}
	u(\vec{b}, \vec{x}) 
	&= \chi_{\{b_1 > \max\limits_{j=2,\dots,n} b_j\}} \Big(x_1 - \max_{j=2,\dots,n} b_j\Big) \\
\end{align*}
Since every $\tilde{\beta} \in \mathcal{B}_\delta$ is an increasing function, it satisfies $\max\limits_{j=2,\dots,n} \tilde{\beta}(x_j) = \tilde{\beta}\big(\max\limits_{j=2,\dots,n} x_j\big)$. Using $Y := \max\limits_{j=2,\dots,n} X_j \sim G := F^{n-1} \in C^{0,1}([0,1])$ with derivative $g := G'$, the ex-ante utility $U$ against symmetric strategies $\vec{\tilde{\beta}} = (\tilde{\beta}, \dots, \tilde{\beta})$ can then be reformulated as
\begin{align*}
U(\beta, \vec{\tilde{\beta}}_{-1})
&= \int_{[0,1]^n} \chi_{\{\beta(x_1) > \max\limits_{j=2,\dots,n} \tilde{\beta}(x_j) \}} \big(x_1 - \max_{j=2,\dots,n} \tilde{\beta}(x_j)\big) \diff F(x_1)\cdots \diff F(x_n) \\
&= \int_{[0,1]^n} \chi_{\{\beta(x_1) > \tilde{\beta}(\max\limits_{j=2,\dots,n} x_j) \}} \big(x_1 - \tilde{\beta}(\max_{j=2,\dots,n} x_j)\big) \diff F(x_1) \cdots \diff F(x_n) \\
&= \int_0^1 \int_0^1 \chi_{\{\beta(x) > \tilde{\beta}(y) \}} \big(x - \tilde{\beta}(y)\big) \diff G(y) \diff F(x) .
\end{align*}
This leads to the following expression for the derivative:
\begin{lemma}
The Gateaux-derivative at $(\beta,\vec{\tilde{\beta}}_{-1}) \in \vec{\mathcal{B}}_\delta$ along $d \in V$ is given by
\begin{align}\label{eq:derivative:SP}
	DU(\beta, \vec{\tilde{\beta}}_{-1})[d]
	&= \int_0^1 d(x) \chi_{\{ \beta(x) \in \text{Im}(\tilde{\beta})\}} \big(x - \beta(x)\big) \frac{g(\tilde{\beta}^{-1}(\beta(x)))}{\tilde{\beta}'(\tilde{\beta}^{-1}(\beta(x)))} \diff F(x).
\end{align}
\end{lemma}

\begin{proof}
By definition, we have
\begin{align*}
DU(\beta, \vec{\tilde{\beta}}_{-1})[d] 
&= \lim_{\varepsilon \to 0} \varepsilon^{-1} \left( U(\beta+\varepsilon d, \vec{\tilde{\beta}}_{-1}) - U(\beta, \vec{\tilde{\beta}}_{-1}) \right) \\
&= \lim_{\varepsilon \to 0} \int_0^1 \underbrace{\int_0^1 \varepsilon^{-1} \big(\chi_{\{\beta(x) + \varepsilon d(x) > \tilde{\beta}(y) \}} - \chi_{\{\beta(x) > \tilde{\beta}(y) \}}\big) \big(x - \tilde{\beta}(y)\big) \diff G(y)}_{I(x)\, :=} \diff F(x). 
\end{align*}
The direction $d \in V \subset AC([0,1])$ is uniformly bounded by $\|d\|_{L^\infty(0,1)} < \infty$, so that the integrand $I(x)$ is bounded $F$-a.e. as seen by the following calculation:
\begin{align*}
|I(x)|
&\leq |\varepsilon^{-1}| \int_0^1 \chi_{\{\beta(x) + |\varepsilon| \|d\|_{L^\infty} \geq \tilde{\beta}(y) \geq \beta(x) - |\varepsilon| \|d\|_{L^\infty} \}} \diff G(y) \\
&= |\varepsilon^{-1}| \int_{\tilde{\beta}(0)}^{\tilde{\beta}(1)} \chi_{\{\beta(x) + |\varepsilon| \|d\|_{L^\infty} \geq b \geq \beta(x) - |\varepsilon| \|d\|_{L^\infty} \}} \diff(G\circ\tilde{\beta}^{-1})(b) \\
&\leq |\varepsilon^{-1}| L_{G\circ\tilde{\beta}^{-1}} \big|(\beta(x) + |\varepsilon| \|d\|_{L^\infty}) - (\beta(x) - |\varepsilon| \|d\|_{L^\infty})\big| \\
&= 2  \|d\|_{L^\infty} L_{G\circ\tilde{\beta}^{-1}} \leq 2 \delta^{-1} \|d\|_{L^\infty} \|g\|_{L^\infty(0,1)}.
\end{align*}
The third step follows since $G\circ\tilde{\beta}^{-1}$ is Lipschitz-continuous due to $0 \leq (\tilde{\beta}^{-1})' = (\tilde{\beta}'(\tilde{\beta}^{-1}(z)))^{-1} \leq \delta$ ($F$-a.e.).
Using the Lipschitz-continuity, we also obtain the $F$-a.e. point-wise convergence:
\begin{align*}
I(x)
&= \int_0^1 \varepsilon^{-1} \big(\chi_{\{\beta(x) + \varepsilon d(x) > \tilde{\beta}(y) \}} - \chi_{\{\beta(x) > \tilde{\beta}(y) \}}\big) \big(x - \tilde{\beta}(y)\big) \diff G(y) \\
&= \chi_{\{ \beta(x) { < \tilde{\beta}(1)}\}} \varepsilon^{-1} \int_{\beta(x)}^{\beta(x) + \varepsilon d(x)} \big(x - b\big) \diff(G \circ \tilde{\beta}^{-1})(b)  \qquad\text{$F$-a.e.}\\
&\to \chi_{\{ \beta(x) { < \tilde{\beta}(1)}\}} d(x) \big(x - \beta(x)\big) \frac{g(\tilde{\beta}^{-1}(\beta(x)))}{\tilde{\beta}'(\tilde{\beta}^{-1}(\beta(x)))} \qquad\text{$F$-a.e.\ .}
\end{align*}
By the dominated convergence theorem, we can interchange integration and limit to rewrite the original integral as
\begin{align*}
DU(\beta, \vec{\tilde{\beta}}_{-1})[d]
&= \int_0^1 \lim_{\varepsilon \to 0} I(x) \diff F(x)\\
&= \int_0^1 \chi_{\{ \beta(x) { < \tilde{\beta}(1)}\}} d(x) \big(x - \beta(x)\big) \frac{g(\tilde{\beta}^{-1}(\beta(x)))}{\tilde{\beta}'(\tilde{\beta}^{-1}(\beta(x)))} \diff F(x).
\end{align*}

This expression is obviously linear in $d$ and bounded by $|DU(\beta, \vec{\tilde{\beta}}_{-1})[d]| \leq 2 \delta^{-1}\|g\|_{L^\infty(0,1)}\|d\|_{L^\infty(\mathcal{X},F)}$, i.e., $DU(\beta, \vec{\tilde{\beta}}_{-1}) \in V^*$ for all $\beta, \tilde{\beta} \in \mathcal{B}_\delta$.
\end{proof}

\subsection{Existence and uniqueness}

For symmetric second-price sealed-bid auctions with independent private values, we can show that a unique BNE exists and coincides with the (unique) solution of the VI and of the MVI. Therefore, these notions are equivalent in this particular case, even though we show in the following section that the Gateaux-derivative is not monotone, nor pseudo- nor quasi-monotone.

\begin{lemma}\label{la:SP:BNE}
The symmetric BNE, VI and MVI problems have the ($F$-a.e.) unique solution $\beta^* = \textup{Id}$ in the compact and convex set $\mathcal{B}_\delta \subset V$ for $0 < \delta \leq 1$.
\end{lemma}
\begin{proof}Using the expression \eqref{eq:derivative:SP} for $DU$, the symmetric VI \eqref{ineq:symVI} reads
\begin{align}\label{ineq:symVI:SP}
0 &\geq DU(\vec{\beta^*})[\beta - \beta^*]
= \int_0^1 \big(\beta(x) - \beta^*(x)\big) \big(x - \beta^*(x)\big) \frac{g(x)}{{\beta^*}'(x)} \diff F(x) 
\end{align}
for all $\beta \in \mathcal{B}_\delta$. Obviously, $\beta^* = \textup{Id}$ satisfies the VI, and $\textup{Id} \in \mathcal{B}_\delta$ for $\delta \leq 1$. We further show that this is the only solution of the VI. Let $\beta^* \in \mathcal{B}_\delta$ be any solution of the VI and consider $\beta = \textup{Id}$. Then, \eqref{ineq:symVI:SP} becomes
\begin{align*}
\int_{\mathcal{X}}  \big|x - \beta^*(x)\big|^2 \frac{g(x)}{{\beta^*}'(x)} \diff F(x) \leq 0 .
\end{align*}
Since $N = \{ x \in \mathcal{X} \ | \ g(x) = \frac{\diff}{\diff x}((F(x))^{n-1}) = 0 \}$ is a set of measure zero with respect to $F$, and ${\beta^*}' \geq \delta > 0$, this yields $\beta^*(x) = x$ for $F$-a.e. $x \in \mathcal{X}$.\\
On the other hand, the MVI \eqref{ineq:symMVI} for arbitrary $\beta,\tilde{\beta} \in \mathcal{B}_\delta$ reads
\begin{align*}
0 &\geq DU(\beta, \vec{\tilde{\beta}}_{-1})[\beta - \beta^*] 
= \int_0^1 \chi_{\{ \beta(x) {< \tilde{\beta}(1)}\}} \big(\beta(x) - \beta^*(x)\big) \big(x - \beta(x)\big) \frac{g(\tilde{\beta}^{-1}(\beta(x)))}{\tilde{\beta}'(\tilde{\beta}^{-1}(\beta(x)))} \diff F(x).
\end{align*}
Obviously, $\beta^* = \textup{Id}$ satisfies the MVI. 
Therefore,  $\beta^* = \textup{Id}$ is the only strategy satisfying the necessary and sufficient condition for a BNE.
\end{proof}

\subsection{Monotonicity}

Gradient-based learning for symmetric strategies uses the gradient operator $DU(\beta, \vec{\tilde{\beta}}_{-1})$ with $\beta = \tilde{\beta}$. Therefore, we are particularly interested whether the operator $DU(\vec{\beta})$ is (quasi-)monotone in $\beta \in \mathcal{B}_\delta$.
This condition would ensure convergence for extra-gradient methods \citep{khanh2016ModifiedExtragradientMethod}.
However, we show that even in the most simple setting with two bidders ($n=2$) and uniform priors ($F = \textup{Id}$), the operator $DU$ turns out to be neither monotone nor pseudo- nor quasi-monotone.

The operator $DU$ is monotone if it satisfies
\begin{align}\label{ineq:monotone}
\big(DU(\vec{\tilde{\beta}}) - DU(\vec{\beta})\big)[\tilde{\beta} - \beta] \leq 0 \qquad\forall \beta, \tilde{\beta} \in \mathcal{B}_\delta .
\end{align}
For pseudo-monotonicity we require \citep{khanh2016ModifiedExtragradientMethod}
\begin{align}\label{ineq:pseudomonotone}
DU(\vec{\beta})[\tilde{\beta} - \beta] \leq 0 \quad \Rightarrow \quad DU(\vec{\tilde{\beta}})[\tilde{\beta} - \beta] \leq 0, \qquad\forall \beta, \tilde{\beta} \in \mathcal{B}_\delta ,
\end{align}
while quasi-monotonicity requires \eqref{ineq:pseudomonotone} with strict inequality on the left-hand side.
Note that monotonicity implies pseudo-monotonicity, which, in turn, implies quasi-monotonicity.

\begin{proposition}\label{prop:nonmonotone:SP}
The operator $DU$ is neither monotone, nor pseudo- nor quasi-monotone (for $\delta < \frac{9}{100}$, $F = \textup{Id}$ and $n = 2$).
\end{proposition}

\begin{proof}Plugging \eqref{eq:derivative:SP} into the quasi-monotonicity condition \eqref{ineq:pseudomonotone} yields
\begin{align*}
\int_0^1 \frac{\big(\tilde{\beta}(x) - \beta(x)\big) \big(x - \beta(x)\big) }{\beta'(x)} \diff x < 0
\ \ \Rightarrow \ \ 
\int_0^1 \frac{\big(\tilde{\beta}(x) - \beta(x)\big) \big(x - \tilde{\beta}(x)\big)}{\tilde{\beta}'(x)} \diff x \leq 0 .
\end{align*}
A counterexample is given by the piece-wise linear and continuous functions
\begin{align*}
\beta(x) &= \frac{61 x}{100} ,&
\tilde{\beta}(x) &= \begin{cases} x & \text{if } 0 \leq x \leq \frac{1}{3} ,\\
	\frac{9 x}{100} + \frac{91}{300} & \text{if } \frac{1}{3} < x \leq \frac{2}{3} ,\\
	\frac{63 x}{100} - \frac{17}{300} & \text{if } \frac{2}{3} < x \leq 1,
\end{cases}
\end{align*}
which yield a negative integral on the left-hand side, but a positive one on the right-hand side.
Therefore, $DU$ is not quasi-monotone and consequently neither pseudo-monotone nor monotone.
\end{proof}
Note that even though this second-price auction leads to a non-monotone VI, a Minty-type solution exists, which ensures the convergence of projection type methods \citep{Song2020OptDualExtrapolation, huang2023beyond}, see also Appendix \ref{sec:appendix}.

\section{First-price sealed-bid auction}
For symmetric first-price sealed-bid auctions with risk-neutral bidders, we have
\[
u(\vec{b}, \vec{x}) = \chi_{\{b_1 > \max\limits_{j=2,\dots,n} b_j\}} \left(x_1 - b_1\right)
\]
As for the second-price auction, using $Y := \max_{j=2,\dots,n} X_j \sim G := F^{n-1}$, the ex-ante utility $U$ against symmetric strategies $\vec{\tilde{\beta}} = (\tilde{\beta}, \dots, \tilde{\beta})$ can be reformulated as
\begin{align*}
U(\beta, \vec{\tilde{\beta}}_{-1})
&= \int_{[0,1]^n} \chi_{\{\beta(x_1) > \max\limits_{j=2,\dots,n} \tilde{\beta}(x_j) \}} \big(x_1 - \beta(x_1)\big) \diff F(x_1) \cdots \diff F(x_n) \\
&= \int_0^1 \int_0^1 \chi_{\{\beta(x) > \tilde{\beta}(y) \}} \big(x - \beta(x)\big) \diff G(y) \diff F(x) \\
&= \int_0^1 \big(x - \beta(x)\big) \int_0^1 \chi_{\{\beta(x) > \tilde{\beta}(y)\}} \diff G(y) \diff F(x) .
\end{align*}
The Gateaux-derivative at $(\beta,\vec{\tilde{\beta}}_{-1}) \in \vec{\mathcal{B}}_\delta$ along $d \in V$, can be computed as:
\begin{align*}
DU&(\beta,\vec{\tilde{\beta}}_{-1})[d]
= \lim_{\varepsilon \to 0} \varepsilon^{-1} \big(U(\beta+\varepsilon d, \vec{\tilde{\beta}}_{-1}) - U(\beta, \vec{\tilde{\beta}}_{-1})\big) \\
&= \lim_{\varepsilon \to 0} \frac{1}{\varepsilon} \int_0^1 \big(x - \beta(x)- \varepsilon d(x) \big) \int_{\{\beta(x) + \varepsilon d(x) > \tilde{\beta}(y)\}} \diff G(y) 
- \big(x - \beta(x)\big) \int_{\{\beta(x) > \tilde{\beta}(y)\}} \diff G(y) \diff F(x) \\
&= \lim_{\varepsilon \to 0} \int_0^1 \bigg[ \big(x - \beta(x)\big) \frac{1}{\varepsilon} \int_0^1 \chi_{\{\beta(x)+\varepsilon d(x) > \tilde{\beta}(y)\}} - \chi_{\{\beta(x) > \tilde{\beta}(y)\}} \diff G(y) \\&\hspace{2cm}
- d(x) \int_0^1 \chi_{\{\beta(x)+\varepsilon d(x) > \tilde{\beta}(y)\}} \diff G(y) \bigg] \diff F(x) .
\end{align*}
Analogously to the derivation in the previous section, we obtain for $\tilde{\beta} \in \mathcal{B}_\delta$ that the first term converges to $\chi_{\{\beta(x) {< \tilde{\beta}(1)}\}}(x-\beta(x)) d(x) (G \circ \tilde{\beta}^{-1})'(\beta(x))$, while the second one yields $\chi_{\{\beta(x) {< \tilde{\beta}(1)}\}} d(x) G(\tilde{\beta}^{-1}(\beta(x)))$, such that we obtain
\begin{multline}\label{eq:derivative:FP}
DU(\beta,\vec{\tilde{\beta}}_{-1})[d]
= \int_0^1 d(x) \chi_{\{\beta(x) {< \tilde{\beta}(1)}\}} \bigg[ \big(x - \beta(x)\big) \frac{g(\tilde{\beta}^{-1}(\beta(x)))}{\tilde{\beta}'(\tilde{\beta}^{-1}(\beta(x)))}- G(\tilde{\beta}^{-1}(\beta(x))) \bigg] \diff F(x) .
\end{multline}
Hence, the operator $DU(\beta, \vec{\tilde{\beta}})$ is linear and in $V^*$.

\subsection{Uniqueness of BNE and non-existence of MVI}

For symmetric first-price sealed-bid auctions with independent private values, we can show that the unique BNE  coincides with a solution of the VI (which is unique in the interior of $\mathcal{B}_\delta$). Even in the simple case of two bidders ($n=2$) with uniform priors ($F=\textup{Id}$), no solution to the MVI exists. So, in contrast to second-price auctions, these notions are different for first-price auctions.

\begin{lemma}\label{la:FP:BNE}
Assume that $f = F'$ is Lipschitz-continuous and satisfies
\[ \delta_0 := \frac{\inf_{x \in [0, 1]} f(x)}{\sup_{x \in [0, 1]} f(x)} > 0. \]
For $0 < \delta \leq \delta_0$,
$\beta^*(x) = \frac{1}{G(x)} \int_0^x y \diff G(y)$ is the unique solution to the symmetric VI \eqref{ineq:symVI} in the interior of $\mathcal{B}_\delta$. This solution is the unique symmetric BNE of \eqref{ineq:symBNE}.
\end{lemma}

\begin{proof}
	The unique symmetric BNE is given by $\beta^*(x) = \frac{1}{G(x)} \int_0^x y \diff G(y)$ \citep{krishna2009AuctionTheory, chawla2013auctions}. Next, we show that this is the unique solution of the symmetric VI \eqref{ineq:symVI} in the interior of $\mathcal{B}_\delta$. Using \eqref{eq:derivative:FP}, the symmetric VI \eqref{ineq:symVI} reads
\begin{align}\label{ineq:symVI:FP}
\int_{\mathcal{X}} \Big[ \big(x - \beta^*(x)\big) \frac{g(x)}{{\beta^*}'(x)} - G(x) \Big] \big(\beta(x) - \beta^*(x)\big) \diff F(x)
\leq 0 \qquad \forall \beta \in \mathcal{B}_\delta.
\end{align}
A solution in the interior of $\mathcal{B}_\delta$ must satisfy
\begin{align*}
\big(x - \beta^*(x)\big) \frac{g(x)}{{\beta^*}'(x)} - G(x) = 0 \qquad\text{$F$-a.e.}
\end{align*}
This ODE can be rearranged to $\frac{\diff}{\diff x} (G(x)\beta^*(x)) = x g(x)$. Since $f$ is Lipschitz-continuous, the same holds for the right-hand side $xg(x) = x f(x) (F(x))^{n-2}$. By the Picard--Lindel\"of theorem the unique solution in the interior of $\mathcal{B}_\delta$ is given by $\beta^*(x) = \frac{1}{G(x)} \int_0^x y \diff G(y)$ due to $\beta(0) = 0$ for any $\beta \in \mathcal{B}_\delta$. Note that $\beta^* \in \mathcal{B}_\delta$ since
\begin{align*}
{\beta^*}'(x)
= \frac{g(x)}{(G(x))^2} \int_0^x G(y) \diff y
&= \frac{f(x) \int_0^x (F(y))^{n-1} \diff y}{\int_0^x f(y) (F(y))^{n-1} \diff y} \\
&\geq \frac{\inf_{z \in [0,1]} f(z) \int_0^x (F(y))^{n-1} \diff y}{\int_0^x \sup_{z \in [0,1]} f(z) (F(y))^{n-1} \diff y} 
= \frac{\inf_{z \in [0,1]} f(z)}{\sup_{z \in [0,1]} f(z)}
= \delta_0 \geq \delta > 0.
\end{align*}
Then, $\beta^* \in \mathcal{B}_\delta$ satisfies \eqref{ineq:symVI:FP}.
\end{proof}

\begin{remark}
Note that Lemma \ref{la:FP:BNE} holds for all $0 < \delta \leq \delta_0$. {Hence,} a limit argument shows that the BNE is the unique solution to the symmetric VI \eqref{ineq:symVI:FP} in the interior of the class of uniformly increasing functions $B_{0+} := \bigcup_{\delta > 0} B_\delta$. If a solution $\beta$ to the VI at the boundary $\partial B_{0+}$ exists, {its} derivative $\beta'$ {must approach} zero at some point $x \in [0,1]$, such that the expression $DU(\vec{\beta})$ might be ill-defined.
{In particular, this implies that a gradient-based learning algorithm must reach the BNE if it does converge (within $B_{0+}$).}
\end{remark}

\begin{lemma}
In the case of two bidders with uniform priors, i.e., for $n=2$ and $F=\textup{Id}$, the unique BNE $\beta^*(x) = \frac{x}{2}$ according to Lemma~\ref{la:FP:BNE} does not satisfy the symmetric MVI \eqref{ineq:symMVI}. In particular, {the condition is also not satisfied locally for any open neighborhood of the BNE (for $\delta \leq \frac{1}{5}$).} 
\end{lemma}

\begin{proof}Using \eqref{eq:derivative:FP}, the symmetric MVI \eqref{ineq:symMVI} reads
\begin{align}\label{ineq:symMVI:FP}
\int_0^1 \Big[ \frac{x - \beta(x)}{{\beta}'(x)} - x \Big] \big(\beta(x) - \beta^*(x)\big) \diff F(x)
\leq 0 \qquad \forall \beta \in \mathcal{B}_\delta.
\end{align}
Using $\beta^*(x) = \frac{x}{2}$ and the continuous, piece-wise linear and strictly increasing bid function
\[
\beta(x) = \begin{cases}
\frac{x}{2} & \text{for}\ x \leq \frac{n}{n+2} ,\\
\frac{n}{2(n+2)} + \frac{4}{5}\left(x - \frac{n}{n+2}\right) & \text{for}\ \frac{n}{n+2} < x \leq \frac{n+1}{n+2} ,\\
\frac{n}{2(n+2)} + \frac{4}{5(n+2)} + \frac{1}{5} \left(x - \frac{n+1}{n+2}\right) & \text{for}\ \frac{n+1}{n+2} < x ,
\end{cases}
\]
for arbitrary $n \in \mathbb{N}_0$, we obtain on the left{-hand side} of \eqref{ineq:symMVI:FP} a positive value
which then contradicts \eqref{ineq:symMVI:FP}. In particular, this variational stability condition is even violated locally, since $\beta \to \beta^*$ in $V$ as $n \to \infty$.
\end{proof}

\subsection{Monotonicity}

As before, we are interested in the situation $\beta = \tilde{\beta}$ along which gradient-based learning takes place, and {study} whether the operator $DU$ is (quasi-)monotone in $\mathcal{B}_\delta$.
Again, we show that even in the most simple setting of two bidders ($n=2$) with uniform priors ($F = \textup{Id}$), the operator $DU$ turns out to be neither monotone, nor pseudo- nor quasi-monotone.

\begin{proposition}
The operator $DU$ is neither monotone, nor pseudo- nor quasi-monotone (for $0 < \delta \leq \frac{1}{10}$).
\end{proposition}
\begin{proof}
For $F(x) = x$ and $n = 2$, and using \eqref{eq:derivative:SP}, the quasi-monotonicity condition \eqref{ineq:pseudomonotone} reads
\begin{multline}\label{ineq:qmono:FP}
\int_0^1 \big(\tilde{\beta}(x) - \beta(x)\big) \Big( \tfrac{x - \beta(x)}{\beta'(x)} - x \Big) \diff x < 0
\Rightarrow \quad
\int_0^1 \big(\tilde{\beta}(x) - \beta(x)\big) \Big( \tfrac{x - \tilde{\beta}(x)}{\tilde{\beta}'(x)} - x \Big) \diff x \leq 0 .
\end{multline}
A counterexample is given by the piece-wise linear and continuous functions
\begin{align*}
\beta(x) &= \frac{61 x}{100} ,&
\tilde{\beta}(x) = \begin{cases} x & \text{for } x \leq \frac{1}{3} ,\\
	\frac{x}{10} + \frac{3}{10} & \text{for } \frac{1}{3} < x \leq \frac{2}{3} ,\\
	\frac{63 x}{100} - \frac{4}{75} & \text{for } \frac{2}{3} < x,
\end{cases}
\end{align*}
which yield a negative value for the left-hand side of \eqref{ineq:qmono:FP}, but a positive value for the right-hand side of \eqref{ineq:qmono:FP}. Therefore, $DU$ is not quasi-monotone and consequently neither pseudo-monotone nor monotone. Note that $\tilde{\beta}' \geq \frac{1}{10} \geq \delta$.
\end{proof}

\section{Numerical Analysis}
In order to better understand the violations of the Minty condition in first-price auctions, we analyze symmetric learning in the space of piecewise linear functions. More specifically, we analyze bid functions consisting of two linear functions in a first- and second-price sealed-bid auction with two bidders having a uniform prior distribution. These bid functions are parametrized by the slopes  $b_1, b_2$ of the two pieces, i.e., 
\begin{equation}
\beta(x) = \begin{cases}
	b_1 x &\text{if } x \leq \tfrac{1}{2} \\
	\tfrac{b_1}{2} + b_2 (x-\tfrac 1 2) &\text{if } \tfrac{1}{2} < x.
\end{cases}    
\end{equation}
This low-parameter environment allows us to examine the resulting vector field and those areas where the Minty condition is violated (see Figure \ref{Minty}).
\begin{figure}
\centering
\includegraphics[width=0.8\textwidth]{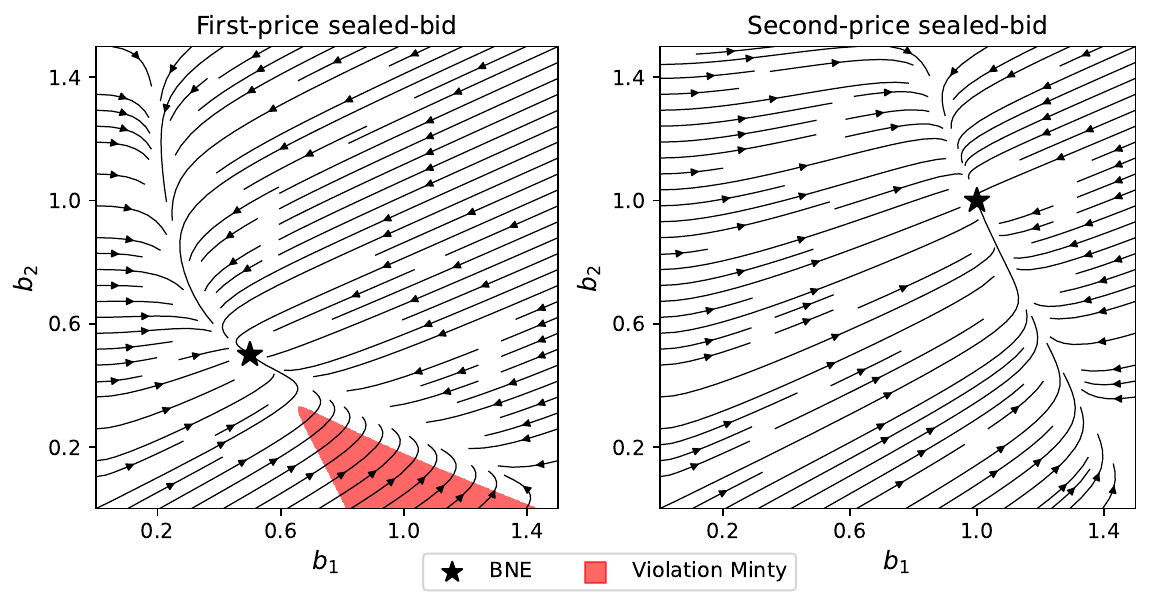} \caption{Gradients for piecewise linear bid functions with two pieces.}
\label{Minty}
\end{figure}
The unique Bayes--Nash equilibrium (marked with a star) is to have a slope of $\tfrac 1 2$ for the first-price auction and a slope of $1$ for the second-price auction for both linear functions. The areas colored red mark where the Minty condition is violated. The figure shows that violations of the Minty condition, a sufficient condition for convergence, are without loss in this analysis. The laminar flows take a turn, but no matter with which parameter combination one starts, the gradient flow always leads to the Bayes--Nash equilibrium. Traversing the red area means that the strategies analyzed on the way move away from the Bayes--Nash equilibrium, but they move closer again once the algorithm steps outside this area.\footnote{Let us note that the utility of the agents increases with lower slopes. Utility does not increase along the trajectories of the gradient such that the utility function is no Lyapunov function of the dynamical system even in the symmetric case. } 

\section{Conclusions}

Learning in games has received much recent attention in the literature. It is well known, that learning algorithms do not always converge to an equilibrium in games, but they do converge in some types such as potential games. Recent advances in equilibrium learning showed that learning algorithms converge in a wide variety of auction games. The reasons for these observations are not well understood. We draw on the connection between auction games and infinite-dimensional variational inequalities, which has not been explored so far. In particular, there are sufficient conditions for which it has been shown that independent optimization algorithms find a solution to the variational inequality.  

Monotonicity can be seen as a generalization of convexity in optimization, and it provides a condition for first-order optimization methods to converge to a solution of the variational inequality. Our analysis shows that neither the second- nor the first-price auctions are monotonous. There are even counterexamples for the weaker pseudo- and quasi-monotonicity conditions. More recent literature on non-monotone variational inequality uses the Minty condition to show the convergence of extragradient algorithms (cf. Appendix \ref{sec:appendix}). In the first-price auction, this condition is also not satisfied, even when assuming a uniform prior. 

However, both auctions have a unique solution {to the variational inequality} (in case of first-price auctions up to the boundary), which means, that every gradient-based learning algorithm must attain the {Bayes--Nash equilibrium (BNE)} if it does converge. Regarding this problem as a dynamical system provides a potential remedy. Since the BNE is the only stationary point (in the interior), the gradient dynamics either converge to the BNE, converge to a limit cycle, or diverge. In general, finding limit cycles of dynamical systems is a PSPACE-complete problem \citep{papadimitriou2015computational}. Our numerical explorations illustrate that the gradient dynamics do converge in the first- and second-price auctions. Whether limit cycles or divergence are possible in other auction games, remains an open research question to be studied in the future. 

\section*{Acknowledgments}
This project was funded by the Deutsche Forschungsgemeinschaft (DFG, German Research Foundation) under Grant No.  BI 1057/9; WO 671/11; and Project Number 277991500.


\bibliographystyle{informs2014} 



\newpage
\begin{appendix}
	
	\section{Minty condition and extragradient algorithms}\label{sec:appendix}
	
	In the following, we consider the Minty condition for variational inequalities, which was recently shown to be sufficient for a number of algorithms \citep{strodiot2016class, Song2020OptDualExtrapolation}. This condition is also referred to as the {Minty VI or dual VI} of the Stampacchia-type VI \citep{ye2022infeasible}. In particular, \cite{Song2020OptDualExtrapolation} show that when a Minty-type solution exists, optimistic dual extrapolation converges to a Stampacchia-type solution of the VI. Here, we discuss this approach for the symmetric second-price sealed-bid auction, which has a MVI solution as shown in Lemma \ref{la:SP:BNE}. In particular, this implies that at least the rather expensive optimistic dual extrapolation provably finds the BNE of symmetric second-price auctions with independent private values.
	
	To apply the results of \cite{Song2020OptDualExtrapolation}, we first need to reformulate the problem in a Hilbert space setting. To this end, we use the Hilbert space $\mathcal{H} = H^1((0,1); F)$ with inner product
	\[
	(v,w)_{\mathcal{H}} := \int_0^1 v(x)w(x) + v'(x)w'(x) \diff F(x)
	\]
	and induced norm $\|v\|_{\mathcal{H}} := \sqrt{(v,v)_{\mathcal{H}}}$. This norm satisfies Assumption 2 of \cite{Song2020OptDualExtrapolation}.
	We consider the closed and convex set $\mathcal{W}_\delta := \{ h \in \mathcal{H} \ : \ 0 < \delta \leq h' \ \text{$F$-a.e., and} \ 0 \leq h \leq 1 \ \text{$F$-a.e.} \}$ instead of $\mathcal{B}_\delta$ before, since we replaced the Banach space $V$ by the Hilbert space $\mathcal{H}$.
	The operator $DU : \mathcal{W}_\delta \subset \mathcal{H} \to \mathcal{H}^*$ is Lipschitz-continuous
	with Lipschitz-constant $L_{DU} \leq 2 \delta^{-2} \|g\|_{L^\infty}$ since
	\begin{align*}
		\big| \big(DU(\vec{\beta}) - DU(\vec{\tilde{\beta}})\big)[d] \big|
		&= \left|\int_0^1 d(x) \Big[ \frac{x - \beta(x)}{\beta'(x)} - \frac{x - \tilde{\beta}(x)}{\tilde{\beta}'(x)} \Big] g(x) \diff F(x)\right| \\
		&\leq \|d\|_{L^2}\|g\|_{L^\infty} \left\| \frac{\big(x - \beta(x)\big)\tilde{\beta}'(x) - \big(x - \tilde{\beta}(x)\big)\beta'(x)}{\beta'(x)\tilde{\beta}'(x)} \right\|_{L^2} \\
		&\leq \|d\|_{L^\infty}\|g\|_{L^\infty} \left\| \frac{\tilde{\beta}(x) - \beta(x)}{\beta'(x)}\right\|_{L^2} \\&\quad
		+ \|d\|_{L^\infty}\|g\|_{L^\infty} \left\|\frac{\big(x - \tilde{\beta}(x)\big)\big(\tilde{\beta}'(x)-\beta'(x)\big)}{\beta'(x)\tilde{\beta}'(x)} \right\|_{L^2} \\
		&\leq \|d\|_{L^\infty}\|g\|_{L^\infty} \left( \delta^{-1} \| \beta - \tilde{\beta} \|_{L^2} + \delta^{-2} \| \beta'-\tilde{\beta}'\|_{L^2} \right) \\
		&\leq 2 \delta^{-2} \|d\|_{L^\infty} \|g\|_{L^\infty} \| \beta - \tilde{\beta} \|_{\mathcal{H}} 
	\end{align*}
	for ${0 <} \delta \leq 1$. Hence, Assumption 1 of \cite{Song2020OptDualExtrapolation} is satisfied. Finally, Assumption 2  of \cite{Song2020OptDualExtrapolation} requires the existence of {a} solution to the MVI \eqref{ineq:symMVI}, which was shown in Lemma \ref{la:SP:BNE} (Note that $\mathcal{H} = H^1(0,1; F) \subset V = W^{1,1}(0,1; F)$ and $\beta^* = \textup{Id} \in \mathcal{H}$).
	Following the argumentation \cite{Song2020OptDualExtrapolation}, now applied to the infinite-dimensional $\mathcal{H}$ instead of $\mathbb{R}^n$, we obtain the convergence of the optimistic dual extrapolation
	\cite[Algorithm 1]{Song2020OptDualExtrapolation}, given here in Algorithm~\ref{alg:OptDEA}. In particular, we obtain:
	\begin{proposition}
		After $K$ iterations, Algorithm~\ref{alg:OptDEA} returns a $\tilde{\beta}_K$ such that
		\begin{align*}
			\sup_{\begin{smallmatrix} \beta \in \mathcal{W}_\delta \\ \|\tilde{\beta}_K - \beta\|_{\mathcal{H}} \leq D \end{smallmatrix}} \big|DU(\vec{\tilde{\beta}}_K)[\tilde{\beta}_K-\beta]\big| \leq \sqrt{8} \left(1 + \alpha^{-1}\right) D L_{DU} \| \beta_0 - \beta^* \|_{\mathcal{H}} K^{-1/2} .
		\end{align*}
	\end{proposition}
	The parameter $\alpha$ can be freely chosen in the interval $(0, \frac{1}{4\sqrt{2}}]$. The best error constant is achieved for the maximal $\alpha = \frac{1}{4\sqrt{2}}$.
	Note that in Algorithm~\ref{alg:OptDEA} the proximal mapping $P_v : \mathcal{H}^* \to \mathcal{H}$ for $v \in \mathcal{H}$ is given by
	\[
	P_v(E) := \argmin_{z \in \mathcal{W}_\delta} \Big( E[z] + \frac{1}{2} \|z - v\|_{\mathcal{H}}^2 \Big) .
	\]
	This requires solving a \emph{monotone} VI on $\mathcal{W}_\delta$ for each evaluation.
	\newpage
	\begin{algorithm}
		\caption{Optimistic dual extrapolation algorithm (\cite{Song2020OptDualExtrapolation})}
		\label{alg:OptDEA}
		\algrenewcommand\algorithmicrequire{\textbf{Input:}}
		\begin{algorithmic}[1]
			\Require{Lipschitz constant $L_{DU} > 0$ and parameter $0 < \alpha \leq \frac{1}{4\sqrt{2}}$.}
			\State $\beta_0 = z_0 \in W_\delta$, $g_0 = 0 \in \mathcal{H}^{-1}$
			\For{ $k = 1, 2, 3, \dots, K$ }
			\State $\beta_k = P_{z_{k-1}}\big( -\frac{\alpha}{L} DU(\vec{\beta_{k-1}}) \big)$
			\State $g_k = g_{k-1} - \frac{\alpha}{L} DU(\vec{\beta_k})$
			\State $z_k = P_{\beta_0}\big( g_k \big)$
			\EndFor
			\State $\tilde{\beta}_K = \argmin_{\beta_k \, : \, 1 \leq k \leq K} \big( \| \beta_k - z_{k-1} \| + \| \beta_{k-1} - z_{k-1} \| \big)$
		\end{algorithmic}
	\end{algorithm}
	\vfill
	
\end{appendix}

\end{document}